\RequirePackage{xcolor} 
\documentclass[journal,twoside,web]{ieeecolor}
\usepackage{lcsys}

\usepackage{etoolbox} 
\makeatletter
\@ifundefined{color@begingroup}%
  {\let\color@begingroup\relax
   \let\color@endgroup\relax}{}%
\def\fix@ieeecolor@hbox#1{%
  \hbox{\color@begingroup#1\color@endgroup}}
\patchcmd\@makecaption{\hbox}{\fix@ieeecolor@hbox}{}{\FAILED}
\patchcmd\@makecaption{\hbox}{\fix@ieeecolor@hbox}{}{\FAILED}

\usepackage{epsfig}
\usepackage{url}
\usepackage{amsmath}
\usepackage{amsfonts}

\usepackage{amsthm}
\usepackage{amssymb}
\usepackage{graphicx}
\usepackage{balance} 
\usepackage{algorithm2e}
\usepackage{algorithmic}
\usepackage{bm}

\newcommand{\bd}{\mathbf}
\newcommand{\C}{\mathcal{C}}
\DeclareMathOperator{\Tr}{Tr}
\DeclareMathOperator{\conv}{Conv}
\newcommand{\bma}{\begin{bmatrix}}
\newcommand{\ebma}{\end{bmatrix}}
\newcommand{\Idq}{\bd{I}_{\mathrm{dq}}}
\newcommand{\Vdq}{\bd{V}_{\mathrm{dq}}}
\newcommand{\oVdq}{\overline{\mathbf{V}}_\mathrm{dq}}
\newcommand{\oVm}{\overline{V}_\mathrm{dq}^2}
\newcommand{\oIdq}{\overline{\mathbf{I}}_\mathrm{dq}}

\let\labelindent\relax

\newcommand{\edit}[1]{{#1}} 
\usepackage{enumitem}
\usepackage{graphicx} 
\graphicspath{{figure/}}
\usepackage{float}
\usepackage[caption=false,font=footnotesize]{subfig}
\usepackage{soul}
\usepackage{color}

\usepackage{algorithm,algorithmic}
\usepackage{cite}

\newtheoremstyle{bfnote}%
{}{}%
{\itshape}{}%
{\bfseries}{.}%
{ }%
{\thmname{#1}\thmnumber{ #2}\thmnote{ (#3)}}
\theoremstyle{bfnote}
\newtheorem{thm}{Theorem}

\newtheorem{lemma}{Lemma}

\usepackage{circuitikz}
\usepackage{tikz}

\title{\LARGE \bf Geometry of the Feasible Output Regions of Grid-Interfacing Inverters with Current Limits}

\author{Lauren Streitmatter, Trager Joswig-Jones and Baosen Zhang
\thanks{Department of Electrical and Computer Engineering, University of Washington Seattle, WA 98195, USA \{lstreit, joswitra, zhangbao\}@uw.edu} %
\thanks{The authors are partially supported by the NSF grant ECCS-1942326, the Washington Clean Energy Institute, the Grainger Foundation and the Galloway Foundation.}}

\begin{document}
\maketitle
\thispagestyle{empty}
\pagestyle{empty}

\begin{abstract}
Many resources in the grid connect to power grids via programmable grid-interfacing inverters that can provide grid services and offer greater control flexibility and faster response times compared to synchronous generators. However, the current through the inverter needs to be limited to protect the semiconductor components. Existing controllers are designed using somewhat ad hoc methods, for example, by adding current limiters to preexisting control loops, which can lead to stability issues or overly conservative operations. 

In this paper, we study the geometry of the feasible output region of a current-limited inverter. We show that under a commonly used model, the feasible region is convex. We provide an explicit characterization of this region, which allows us to efficiently find the optimal operating points of the inverter. We demonstrate how knowing the feasible set and its convexity allows us to improve upon existing grid-forming inverters such that their steady-state currents always remain within the current magnitude limit, whereas standard grid-forming controllers can lead to instabilities and violations.

\end{abstract}


\section{Introduction}
Many resources in the electric system, including solar PV, wind, storage, and electric vehicles (EVs), are connected to the grid through power electronic inverters. At the same time, power systems were designed assuming the presence of large spinning machines~\cite{kundur1994power}. As fossil fuel-based generation retires and inverter-based resources (IBRs) grow, understanding whether the latter could successfully replace the former in grid operations is becoming increasingly important. 


A key difference between IBRs and synchronous generators is their ability to handle current during normal and contingent operations~\cite{fan2022review,baeckeland2024overcurrent}. Synchronous generators can typically supply more than 5 to 10 times their rated current with relatively little damage~\cite{elnaggar2013comparison}. In contrast, in order to protect {semiconductor devices}, inverter currents can only marginally exceed their nominal rated currents~\cite{hooshyar2017microgrid,ieee2022ieee}. In recent years, significant attention has been paid to how currents should be limited, for example, see~\cite{paquette2014virtual,zhong2016current,fan2022review,baeckeland2024overcurrent} and the references within. 

However, despite these results, the dispatch signals coming from the system are often somewhat oblivious to the current limits. In most cases, inverters are expected to track some setpoints, for example, current/power for grid following inverters and voltage/power for grid-forming inverters~\cite{pattabiraman2018comparison,christensen2020high,li2022revisiting}. These setpoints are optimized to maximize the efficiency of the overall system, but they may not include the physical constraint of the inverters. Even when the inverter constraints are included, they are often presented in simple forms (e.g. upper and lower bounds on all quantities) which do not match the geometry of the actual feasible operation regions. 

Inverters achieve given setpoints using control loops (typically a current loop and a voltage loop). But if reaching these setpoints leads to a violation of the current rating, the inverter output current would saturate~\cite{baeckeland2024overcurrent}, resulting in a mismatch between the higher-level commands and the physical limit of the device. Depending on how the current limiter is designed, the inverter could become dynamically unstable and could cause severe problems in a grid~\cite{hart2014energy,qoria2020current,joswig2024optimal,joswig2024safe}.  

In this paper, we study the geometry of the feasible operating region of a current-limited inverter. That is, given a current limit and the inverter parameters, what is the region of all achievable outputs? Answering this question is the natural first step towards ensuring the safe operations of inverters and designing controllers that optimize their performances. We show that, under a commonly used model, the feasible region is convex when two of active power, reactive power and square of the voltage magnitudes~\footnote{It is more common to consider the magnitude of voltage and not its square. However, the squared magnitude is much easier to work with.} are considered. We provide an explicit characterization of this region using linear matrix inequalities, so we can efficiently find the optimal operating points of the inverter. 

Because power and voltage magnitude depend quadratically on the current, and the current is a vector in $\mathbb{R}^2$ (in the $\mathrm{dq}$ frame), the feasible region is a quadratic map of a disk (current magnitude is upper bounded). Unlike affine transformations, a quadratic map of a convex set is not, in general, convex. It turns out that the particular structure of the inverter circuit preserves convexity, which we will explain in detail in this paper. \edit{We demonstrate how knowing the feasible set and its convexity allows us to improve upon existing grid-forming inverters such that their steady-state currents always remain within the current magnitude limit, whereas standard grid-forming controllers can lead to instabilities and violations}.

\section{Model and Problem Formulation} \label{sec:model}
\subsection{Inverter Model}
\edit{In this study, we consider a widely used model of a three-phase inverter connected to a fixed grid voltage represented by an infinite bus through an \emph{RL} branch as shown in Figure~\ref{fig:circuit}~\cite{yazdani2010voltage,joswig2024optimal}.} 
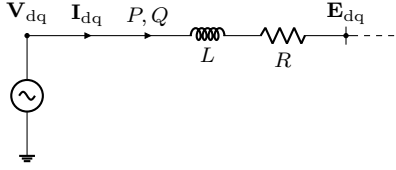
\begin{figure}[ht]
    \begin{center}
\begin{circuitikz} [american voltages,scale=0.8, font=\footnotesize]
\ctikzset{bipoles/length=7mm} 
\draw
    (0,0) to node[tlground]{} (0,0)
    to[sV] (0,2)
    (0,2) node[circ,label=above:$\mathbf{V}_{\mathrm{dq}}$]{}
    to (0,2)--(0.5,2) 
    
    to [short, i=$\mathbf{I}_{\mathrm{dq}}$] (1.5,2)
    
    to [short, i=${P,Q}$] (2.5,2)
    to[L, l_=$L$] (3.5,2)
    
    to[R, l_=$R$] (5,2)

    
    to (5,2) -- (5.5,2) 
    (5.3,1.85) -- (5.3,2.15)
    (5.3,2) node[circ,label=above:$\mathbf{E}_{\mathrm{dq}}$]{}
    (5.6,2) -- (5.7,2)
    (5.8,2) -- (5.9,2)
    (6.0,2) -- (6.1,2);
\end{circuitikz}
\end{center}
    \caption{We model the inverter as a controllable voltage source in the $\mathrm{dq}$-frame connected via an $RL$ filter to an infinite bus.}
    \label{fig:circuit}
\end{figure}
The model assumes a balanced three-phase system and can thus use the direct-quadrature ($\mathrm{dq}$) frame transformation with reference to some angle $\theta$ rotating at constant frequency $\omega$ to describe all the rotating physical quantities in $\mathbb{R}^2$ rather than $\mathbb{R}^3$. We assume the grid-side voltage is of constant magnitude, $E$, and define $\mathbf{E}_\mathrm{dq}$ with respect to the grid-side voltage angle, $\theta$, to get $\mathbf{E_\mathrm{dq}}=(E,0)$. The dynamics of a voltage-source inverter connected to $\mathbf{E_\mathrm{dq}}$ through resistor $R$ and inductor $L$ are governed by 
\begin{equation}
\label{eq:rl_dynamics}
\begin{bmatrix} 
\dot{I}_\mathrm{d} \\\ \dot{I}_\mathrm{q} 
\end{bmatrix} = 
\mathbf{A} \begin{bmatrix} 
I_\mathrm{d} \\\ I_\mathrm{q} 
\end{bmatrix} + \frac{1}{L}
\begin{bmatrix}
    V_{\mathrm{d}}-E\\ V_{\mathrm{q}}
\end{bmatrix},
\end{equation}
where $\mathbf{A} = \begin{bmatrix}
-R/L & \omega \\
-\omega & -R/L
\end{bmatrix}$, and $\mathbf{I_\mathrm{dq}}$ and $\mathbf{V_\mathrm{dq}}$ are the inverter current and voltage, respectively, in the $\mathrm{dq}$-reference frame of the grid \cite{yazdani2010voltage}. \edit{This choice is mainly for notational convenience. We show in Section \ref{sec:module} that our overall approach and main results do not require the knowledge of the grid voltage angle and in fact we can take any angle to be the reference.} We will treat $\Idq$ as the state and $\Vdq$ as the control input. The matrix $\bd A$ being skew-symmetric with equal diagonal elements the the key structural property of the inverter that allows us to derive the main results. 

\edit{To protect the inverter from damaging current levels, the magnitude of the current should be limited~\cite{luo2017advanced,bottrell2013comparison}. Inverter current limits arise from two aspects, thermal limits and inductor saturation~\cite{kassakian2023principles}. Thermal limits are ``soft constraints,'' in the sense that they can be violated for short periods of time. Inductor saturation is a hard constraint and has values that are typically a few times larger than the thermal constraints~\cite{sorensen2013thermal}.} We denote this limit as $I_\mathrm{max}$, and we define the safe (or feasible) current operating region as  
$$\mathcal{I}:=\{ \Idq \mid \|\mathbf{I}_\mathrm{dq}\|_2^2 \leq I_\mathrm{max}^2\}.$$

\subsection{Inverter Output}
Inverters in the grid are usually asked to optimize their output, typically a combination of active power, reactive power, and voltage magnitude.  The natural question becomes how well we can optimize these outputs given the constraint that current needs to stay in the safe region $\mathcal{I}$. 

The inverter's active power, reactive power, and the squared voltage magnitude are all expressed as quadratic functions of the current and the voltage: 
\begin{subequations} \label{eq:pqv}
\begin{align}
    \label{eq:P_inverter}
    P &= \frac{3}{2} \mathbf{I}_\mathrm{dq}^T \mathbf{V}_\mathrm{dq} \\
    \label{eq:Q_inverter}
    Q &= \frac{3}{2} \mathbf{I}_\mathrm{dq}^T \bd J
    \mathbf{V}_\mathrm{dq} \\
    \label{eq:V2_inverter}
    V_\mathrm{dq}^2 &= \mathbf{V}_\mathrm{dq}^T \mathbf{V}_\mathrm{dq},
\end{align}
\end{subequations}
where $\bd J= \bma 0 & 1 \\ -1 & 0 \ebma$. We note that it is more common to consider the voltage magnitude rather than the squared magnitude as we do here. Of course, the two are equivalent from the point of view of achieving a desired setpoint, but the squared form in~\eqref{eq:V2_inverter} is much easier to work with. 

\edit{The quantities in~\eqref{eq:pqv} are defined for all time $t$ since $\Idq$ and $\Vdq$ are functions of $t$, but we are usually interested in optimizing the equilibrium values resulting from the dynamics in~\eqref{eq:rl_dynamics}.} Assuming that \eqref{eq:rl_dynamics} is asymptotically stable (we show this in Section~\ref{sec:control}), we denote the equilibrium of current and voltage as $\overline{\mathbf{I}}_\mathrm{dq}$ and $\overline{\mathbf{V}}_\mathrm{dq}$, respectively. Substituting these into \eqref{eq:pqv}, we get the equilibrium values of $\overline{P}$, $\overline{Q}$ and $\oVm$. \edit{We are interested in two related questions: 1) finding the optimal values of $\overline{P}$, $\overline{Q}$ and $\oVm$ and 2) using these values to improve the performance of existing grid-forming controllers.}

We first look at the question of optimizing the equilibrium values. Let $S_1,S_2\in (\overline{P}, \overline{Q},\overline{V}_\mathrm{dq}^2)$ denote some choice of two out of the three quantities. We are interested in solving the following problem:
\begin{equation}
    \label{eq:opt_control}
    \begin{aligned}
    \text{minimize} \quad & f(S_1,S_2) \\
    \text{subject to} \quad & \oIdq \in \mathcal{I},
    \end{aligned}
\end{equation}
where $f:\mathbb{R}^2 \rightarrow \mathbb{R}$ is some objective function. For example, suppose that we are interested in tracking some setpoints $S_1^*,S_2^*$ as closely as possible. The objective would be $f(S_1,S_2)=\frac12 (S_1-S_1^*)^2+\gamma \frac12 (S_2-S_2^*)^2$ with $\gamma$ being some tradeoff parameter. 

It is not immediately clear whether the problem~\eqref{eq:opt_control} is easy to solve. Even if $f$ is a convex function, the quadratic forms in~\eqref{eq:pqv} make the overall problem not convex in $\oIdq$.

\subsection{Feasible Operating Regions}
We can rewrite \eqref{eq:opt_control} as an optimization over $S_1$ and $S_2$:
\begin{equation}
    \label{eq:opt_S}
    \begin{aligned}
    \text{minimize} \quad & f(S_1,S_2) \\
    \text{subject to} \quad & (S_1,S_2) \in \mathcal{S}
    \end{aligned}
\end{equation}
where $\mathcal{S}$ is the set of all feasible points achieved by currents $\oIdq \in \mathcal{I}$.   In the next section, we show that $\mathcal{S}$ is convex, and consequently, \eqref{eq:opt_S} is a convex optimization problem if $f$ is convex. Using the tracking example again, \eqref{eq:opt_S} with $f(S_1,S_2)=\frac12 (S_1-S_1^*)^2+\gamma \frac12 (S_2-S_2^*)^2$ can be solved efficiently \edit{at real-time}, allowing us to achieve perfect tracking if possible, and the closest setpoints if not.

\section{Geometry of the Feasible Region} \label{sec:module}
In this section, we study the geometry of the feasible set $\mathcal{S}$. The main result is given by Theorem~\ref{thm:main}.
\begin{thm}\label{thm:main}
    Let $(S_1,S_2)$ be a pair of points formed by choosing any two of the three quantities $\overline{P},\overline{Q},\oVm$. Let $\mathcal{S} \in \mathbb{R}^2$ be the set of all achievable points $(S_1,S_2)$ by $\oIdq \in \mathcal{I}$. Then $\mathcal{S}$ is convex.  
\end{thm}
The proof of this theorem is given in the next section. Figure~\ref{fig:opregion} plots what the regions look like for the combinations of $(P,Q), (P,V_\mathrm{dq}^2)$, and $(Q,V_\mathrm{dq}^2)$. \edit{The convexity of the regions do not depend on the values of $R$ and $L$ (as long as they are positive) or $\bd E_\mathrm{dq}$. The exact shape depends on $R$, $L$ and $||\bd E_\mathrm{dq}||_2$, but not the phase angle of $\bd E_\mathrm{dq}$. Therefore, we do not need a PLL to sense the grid voltage angle. We do assume that inverters know their own $RL$ filter specifications or can estimate them, as well as $||E_\mathrm{dq}||$ is measurable~\cite{jayalath2016generalized,kjaer2005review}}.

\begin{figure}[ht]
\centering
    \includegraphics[width=0.7\linewidth]{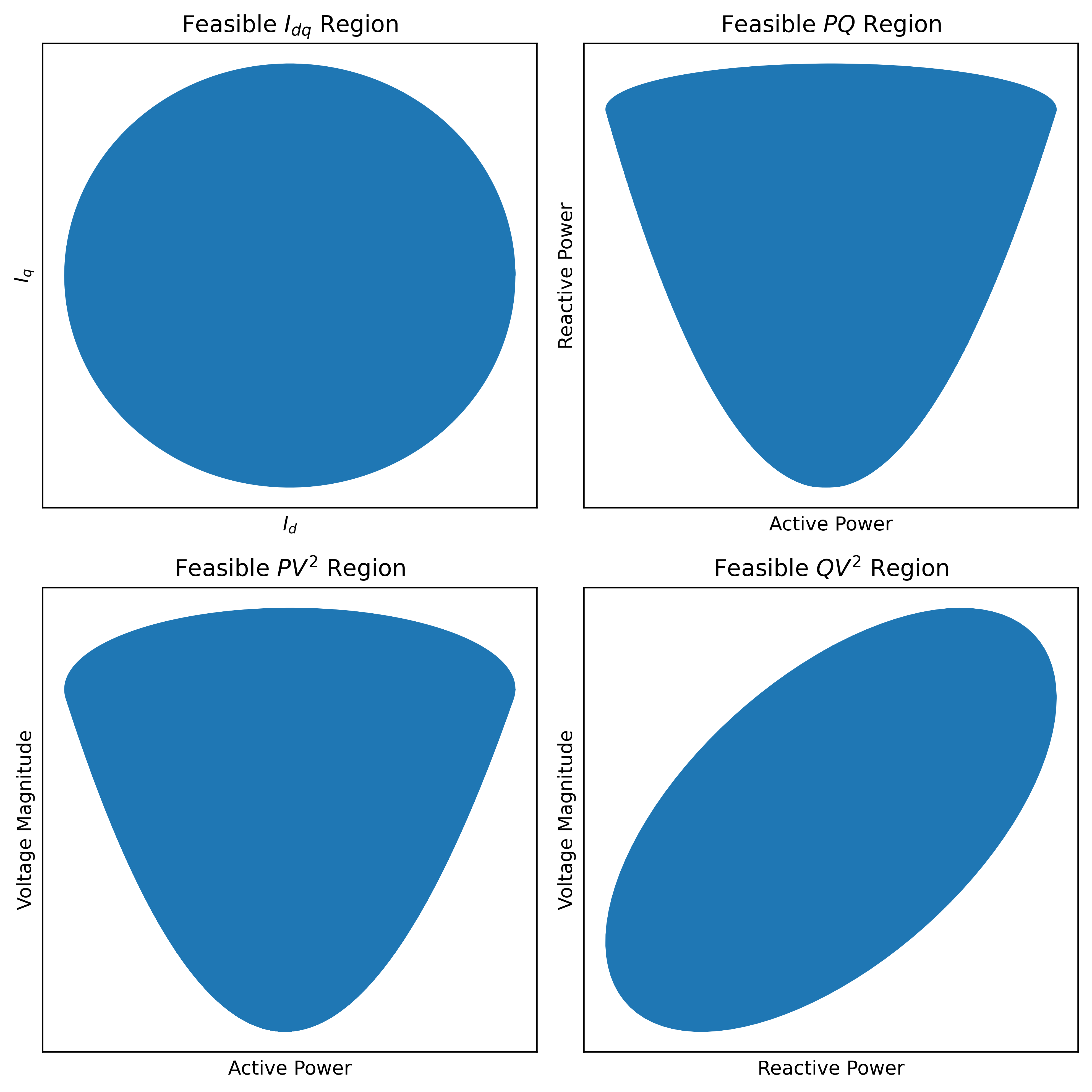}
    \caption{The set of feasible inverter currents ($\mathcal{I})$ (top left) always forms a circle of radius $I_\mathrm{max}$ due to the magnitude constraint on $\mathbf{I}_\mathrm{dq}$. \edit{The shapes of feasible $(P,Q) (\mbox{top right}), (P,V_\mathrm{dq}^2) (\mbox{bottom left})$, and $(Q,V_\mathrm{dq}^2) (\mbox{bottom right})$ regions depend on specific network parameters, namely $RL$ filter impedance values and grid-side voltage magnitude $||\bd E_\mathrm{dq}||_2$}.}
    \label{fig:opregion}
\end{figure}



\subsection{Proof of Theorem~\ref{thm:main}}
The first step in the proof is to write the quantities in \eqref{eq:pqv} just as a function of the current. At equilibrium, the steady state voltage is 
 $$\oVdq = \mathbf{E}_\mathrm{dq} - L\mathbf{A}\oIdq.$$
 Substituting this into \eqref{eq:pqv}, we have
 \begin{subequations} \label{eq:quad}
 \begin{align}
    \label{eq:quad_eq}
    \overline{P} &= \frac{3}{2} \oIdq^T (\mathbf{E}_\mathrm{dq} - L\mathbf{A}\oIdq) \\
    \label{eq:Q_eqbm}
    \overline{Q} &= \frac{3}{2}\oIdq^T \bd J (\mathbf{E}_\mathrm{dq} - L\mathbf{A}\oIdq) \\
    \label{eq:V2_eqbm}
    V_\mathrm{dq}^2 &= (\mathbf{E}_\mathrm{dq} - L\bd A\oIdq)^T (\mathbf{E}_\mathrm{dq} - L\bd A\oIdq).
    \end{align}
\end{subequations}

Using the fact that for any square matrix $\bd M$ and vector $\bd z$, $\bd z^T \bd M \bd z=\bd z^T \frac{\bd M+\bd M^T}{2} \bd z$ and $\frac{\bd A+\bd A^T}{2}=\bma -\frac{R}{L} & 0 \\ 0 & -\frac{R}{L} \ebma=-\frac{R}{L} \bd{I}_2$ (where $\bd{I}_2$ is the 2 by 2 identity matrix), \eqref{eq:quad} becomes
 \begin{subequations} \label{eq:quad_norm}
 \begin{align}
    \label{eq:P_quad_norm}
    \overline{P} &= \frac{3}{2} R ||\oIdq||_2^2+\frac{3}{2} \bd E_\mathrm{dq}^T \oIdq\\
    \label{eq:Q_eqbm_norm}
    \overline{Q} &= \frac{3}{2} \omega L ||\oIdq||_2^2 +\frac{3}{2} \bd E_\mathrm{dq}^T \bd J \oIdq \\
    \label{eq:V2_eqbm_norm}
    V_\mathrm{dq}^2 &= (R^2+\omega^2 L^2) ||\oIdq||_2^2+2 \bd E_\mathrm{dq}^T \bd A \oIdq+||\bd E_\mathrm{dq}||_2^2.
    \end{align}
\end{subequations}
\edit{Note the grid voltage, $\bd E_\mathrm{dq}$, enters into these equations either as $||\bd E_\mathrm{dq}||_2^2$ or as a product with $\oIdq$. Because we are looking at all possible $\oIdq$ satisfying a magnitude constraint, all possible angle differences between the inverter current (or voltage) and the grid voltage are included, and the shapes of the regions do not depend on the phase angle of the grid voltage.}



Because of the particular form of these equations, Theorem 1 follows directly from the following Lemma. 
\begin{lemma} \label{thm:convexity}
Given linearly independent $\mathbf{a}, \mathbf{b} \in \mathbb{R}^2$  and $\alpha, \beta \in \mathbb{R}_+$, the set \( \mathcal{C} = \left\{ \left( \alpha \mathbf{x}^T \mathbf{x} + \mathbf{a}^T \mathbf{x}, \beta \mathbf{x}^T \mathbf{x} + \mathbf{b}^T \mathbf{x} \right) \mid \mathbf{x}^T \mathbf{x} \leq 1 \right\} \) is convex. 
\end{lemma}

To apply Lemma~\ref{thm:convexity} to our problem, first note that the constant term, $||\bd E_\mathrm{dq}||_2^2$, can be ignored.  
We also note that $\mathbf{a}, \mathbf{b}$ will always be linearly independent: the linear terms in \eqref{eq:P_quad_norm} and \eqref{eq:Q_eqbm_norm}, proportional to $\bd E_\mathrm{dq}$ and $\bd J^T \bd E_\mathrm{dq}$, respectively, are orthogonal. The linear term in \eqref{eq:V2_eqbm_norm} is proportional to $\bd A^T\bd E_\mathrm{dq}$ which is a linear combination of the linear terms in \eqref{eq:P_quad_norm} and \eqref{eq:Q_eqbm_norm}. Thus, each of the three vectors is linearly independent with respect to any one of the other two vectors. 

\begin{proof}
To see this, we first rewrite $\mathcal{C}$ as 
\begin{equation}\label{eqn:C}
\left\{ 
\begin{bmatrix} 
\alpha ||\mathbf{x}||_2^2  \\\ \beta ||\mathbf{x}||_2^2
\end{bmatrix} + 
\begin{bmatrix}
\mathbf{a}^T \\\ \mathbf{b}^T
\end{bmatrix} \mathbf{x} \ \mid \ ||\mathbf{x}||_2^2 \leq 1
\right\}.
\end{equation}
Then, we define a new set $\mathcal{C}'$ as the image of set $\mathcal{C}$ under an affine function given by:
\begin{equation} \label{eqn:C'}
    \mathcal{C}'= \left\{ \begin{bmatrix} 
\mathbf{a}^T \\\ \mathbf{b}^T
\end{bmatrix}^{-1} \bd s \mid \bd s \in \mathcal{C} \right\} =  \left\{ ||\mathbf{x}||_2^2 \ \mathbf{c} + \mathbf{x} \mid ||\mathbf{x}||_2^2 \leq 1 \right\},
\end{equation}
for vector $\bd s=(S_1,S_2)\in \mathcal{C}$ and arbitrary vector $\mathbf{c}\in\mathbb{R}^2$, determined by $$\bd c =\begin{bmatrix} 
\mathbf{a}^T \\\ \mathbf{b}^T
\end{bmatrix}^{-1} \begin{bmatrix} \alpha \\\ \beta\end{bmatrix}.$$ The linear independence of $\bd a, \bd b$ ensures the matrix inverse exists. Since  convexity is preserved under affine transformations \cite{Boyd_Vandenberghe_2004}, it suffices to show the convexity of $\mathcal{C}'$.

We now apply the definition of convexity to prove $\mathcal{C}'$ is convex. We will show for any $\lambda\in[0,1]$ and any $\mathbf{x_1}, \mathbf{x_2}$ with $||\mathbf{x_1}||_2^2 \leq 1, || \mathbf{x_2}||_2^2 \leq 1$, $\exists \bd y \in \mathcal{C}'$ such that 
\begin{multline}
\label{eq:condition1}
    \ \| \mathbf{y}\|_2^2 \ \mathbf{c} + \mathbf{y} = (\lambda \| \mathbf{x_1}\|_2^2 + (1-\lambda) \| \mathbf{x_2}\|_2^2) \ \mathbf{c} + \lambda \mathbf{x_1} + (1-\lambda)\mathbf{x_2}, \\ \mathrm{and} \
    \| \mathbf{y}\|_2^2 \leq 1.
\end{multline}
Note that because $\lambda=0$ and $\lambda =1$ result in the trivial cases of $\bd y=\bd x_1$ and $\bd y=\bd x_2$, we assume that $0<\lambda <1$. We recognize from \eqref{eq:condition1} that $\mathbf{y}$ must be of the form:
\begin{gather}
\mathbf{y} = \mu \mathbf{c} + \bd z \mathrm{, where}
\label{eq:y} \\ \mu=\lambda \| \mathbf{x_1}\|_2^2 + (1-\lambda) \| \mathbf{x_2}\|_2^2 - \| \mathbf{y}\|_2^2.\label{eq:mu}
\end{gather}
and $\mathbf{z} = \lambda \mathbf{x_1} + (1-\lambda)\mathbf{x_2}$

The existence of $\bd y$ depends on finding $\mu$ that satisfies \eqref{eq:mu} and results in $||\bd y||_2^2\leq1$. The following steps prove mathematically what is shown visually in Figure \ref{fig:quadratics}: at least one solution for $\mu$ to \eqref{eq:mu} will always lie within the bounds on $\mu$ determined by $||\bd y||_2^2\leq1$.

First, we define $f_1(\mu)$ by substituting \eqref{eq:y} into \eqref{eq:mu} and $f_2(\mu)$ by substituting \eqref{eq:y} into $||\bd y||_2^2\leq1$: 
\begin{equation}
    \label{eq:quadratic_mu}
    f_1(\mu) =\| \mathbf{c}\|_2^2 \mu^2 + (2\mathbf{c}^T\mathbf{z}+1)\mu+(\| \mathbf{z}\|_2^2 -\zeta) = 0,
\end{equation}
\begin{equation}
    \label{eq:magnitude_mu}
    f_2(\mu) =\| \mathbf{c}\|_2^2 \mu^2 + (2\mathbf{c}^T\mathbf{z})\mu+(\| \mathbf{z}\|_2^2 -1) \leq 0,
\end{equation}
where $\zeta = \lambda \| \mathbf{x_1}\|_2^2 + (1-\lambda) \| \mathbf{x_2}\|_2^2\leq1$, and $|| \mathbf{z}||_2^2\leq\zeta$ by the triangle inequality for norms. Note that both parabolas open upwards ($||\bd c||^2_2\geq 0$), and the constant terms of \eqref{eq:quadratic_mu} and \eqref{eq:magnitude_mu} are both negative so each equation will have one positive and one negative root. Denoting the roots of $f_1(\mu)$ as $\mu^-$ and $\mu^+$ and the roots of $f_2(\mu)$ as $\underline{\mu}$ and $\overline{\mu}$, we need at least one of $\mu^-$ or $\mu^+$ to fall within the interval $[\underline{\mu},\overline{\mu}]$ as shown in Figure \ref{fig:quadratics}.

We show $\mu^+\in[\underline{\mu},\overline{\mu}]$ always holds by comparing the zeros and derivatives of \eqref{eq:quadratic_mu} and \eqref{eq:magnitude_mu}. Observe first that $f_1(0) =|| \mathbf{z}||_2^2-\zeta$ and $f_2(0)=|| \mathbf{z}\|_2^2 -1$, meaning $f_2(0)\leq f_1(0)\leq0$ because $\zeta \leq1$. Next, for $\mu > 0$, $\frac{\mathrm{d}f_1(\mu)}{\mathrm{d}\mu}=2\| \mathbf{c}\|_2^2 \mu+2\mathbf{c}^T\mathbf{z}+1$ will always be greater than $\frac{\mathrm{d}f_2(\mu)}{\mathrm{d}\mu}=2\| \mathbf{c}\|_2^2 \mu+2\mathbf{c}^T\mathbf{z}$. Geometrically, this means for $\mu\in \mathbb{R}_+$, the quadratic $f_1(\mu)$ always lies above $f_2(\mu)$ with greater slope, so it will reach $f_1(\mu^+)=0$ for a lower (but still positive) value of $\mu$ than $f_2(\overline{\mu})$. Therefore, $\underline{\mu}\leq \mu^+ \leq \overline{\mu}$ will always hold, proving $\bd y$ from \eqref{eq:condition1} exists and set $\mathcal{C}$ is convex.
\end{proof}

\begin{figure}[ht]
\centering
    \includegraphics[width=0.8\linewidth]{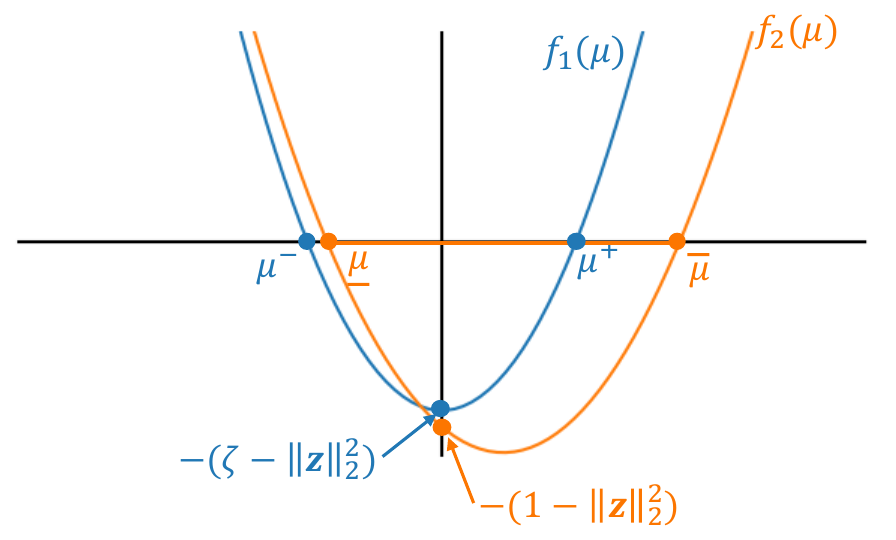}
    \caption{Graphical representation of the quadratic equation $f_1(\mu)$ with its roots and quadratic inequality $f_2(\mu)\leq0$ that corresponds to the interval $\mu \in [\underline{\mu},\overline{\mu}]$. Convexity of set $\mathcal{C}$ depends on at least one of $\mu^-, \mu ^+$ falling within this interval.}
    \label{fig:quadratics}
\end{figure}

Theorem~\ref{thm:main} shows that for any two combinations of $P,Q,V_{dq}^2$, the feasible region is convex. But it does not immediately provide a way to solve \eqref{eq:opt_S}, since it does not directly give an algebraic description of the set. However, as the next theorem shows, we can describe the set using linear matrix inequalities. 
\begin{lemma} \label{lem:W}
    Let $\bd W$ be a 3 by 3 matrix \edit{whose $(i,j)^{th}$ element is denote by $W_{ij}$} and define $\hat{C}$ as 
    \begin{equation*}
    \begin{split}
        \hat{C}=\{ (\Tr(\bd{M}_1 \bd{W}),\Tr(\bd{M}_2 \bd{W}) | W_{11}+W_{22}\leq 1, & W_{33}=1, \\
        & \bd{W} \succeq 0 \},
    \end{split}
    \end{equation*}
    \edit{where $\zeta_1,\zeta_2$ are arbitrary constants and $\bd{I}_2$ is the 2 by 2 identity matrix and}
    \begin{equation} \label{eq:M12} 
    \bd{M}_1=\begin{bmatrix} \bd{I}_2 & \frac12 \bd{a} \\ \frac12 \bd{a}^T & \zeta_1 \end{bmatrix} \mbox{ and } \bd{M}_2=\begin{bmatrix} \bd{I}_2 & \frac12 \bd{b} \\ \frac12 \bd{b}^T & \zeta_2 \end{bmatrix},
    \end{equation}
    Then with $\C$ defined as in Lemma~\ref{thm:convexity}, we have $\C=\hat{\C}.$
\end{lemma}

\begin{proof}

 We first show that $\conv(\C) = \hat{\C}$, where $\conv(\C)$ is the convex hull of $\C$. Observe we can write the equation in~\eqref{eq:quad_norm} in the form of $\bma \bd x \\ 1 \ebma^T \bd M \bma \bd x \\ 1 \ebma$ for some matrix $M$. Then note that the feasible set of $\bd{W}$ is the convex hull of the feasible set for $\bd{x}$:
    \begin{equation*}
    \begin{split}
    & \{W_{11}+W_{22}\leq 1, W_{33}=1, \bd{W} \succeq 0 \}\\
    = &\conv\left(\left\{\begin{bmatrix} \bd x \\ 1 \end{bmatrix}, ||x||_2^2 \leq 1\right\}\right).
    \end{split}
   \end{equation*}
    Then by the linearity of the trace operator,  $\conv(\C) = \hat{\C}$. By Lemma~\ref{thm:convexity}, $\C$ is convex, so we have $\C=\conv(\C)=\hat{\C}$. 
    
\end{proof}


Lemma~\ref{lem:W} shows that we can find optimal steady-state current values by solving
\begin{subequations} \label{eq:sdp}
    \begin{align}
        \min \; & f(\Tr(\bd M_1 \bd W,)\Tr(\bd M_2 \bd W)) \\
        \mbox{s.t. } & W_{11}+W_{22} \leq I_{\max}^2 \\
        & W_{33} =1 \\
        & \bd W \succeq 0,
    \end{align}
\end{subequations}
where $\bd M_1, \bd M_2$ can be chosen to be $\bma \frac32 R \bd{I}_2 & \frac34 \bd E_{\mathrm{dq}} \\ \frac34 \bd E_{\mathrm{dq}}^T & 0  \ebma$ for $P$, $ \bma \frac32 \omega L \bd{I}_2 & \frac34 \bd J^T \bd E_{\mathrm{dq}} \\ \frac34 \bd E_{\mathrm{dq}}^T \bd J & 0 \ebma$ for $Q$ and $\bma (R^2+\omega^2 L^2) \bd{I}_2 & \bd A^T  \bd E_{\mathrm{dq}} \\ \bd E_{\mathrm{dq}}^T \bd A & ||\bd E_{\mathrm{dq}}||_2^2 \ebma $ for $\bd V_{\mathrm{dq}}^2$.

So far we have shown that given a feasible $\bd W$ to \eqref{eq:sdp}, we can find a feasible current $\oIdq$ that achieves the same objective value. However,  the optimal solution to \eqref{eq:sdp} may not be rank 1. We can easily solve a quadratic equation (cf. proof of Lemma~\ref{thm:convexity}) to find the corresponding $\oIdq$, but an interesting observation from simulations is that the solution to \eqref{eq:sdp} is always rank 1 (possibly with a small regularization term on the nuclear norm of $\bd W$). An important part of our future work is to rigorously prove this observation.

\edit{We also remark that \eqref{eq:sdp} can be solved with data that are easy to obtain in practice. The grid voltage magnitude $||\bd E_\mathrm{dq}||_2$ is easy to measure and $R,L$ values can be estimated with existing inverter hardware if not already known. Since the feasible region in \eqref{eq:sdp} is invariant to the rotation of the angle in $\bd E_\mathrm{eq}$, an arbitrary phase angle (e.g., 0) can be chosen.}



\subsection{Inverter Control} \label{sec:control}
\edit{Given the steady-state values $\oIdq$ for \eqref{eq:rl_dynamics}, we need to find a controller such that the current will reach this equilibrium. If we have full information, including the phase angle of the grid voltage, we could use a 
linear feedback control on $\mathbf{V}_\mathrm{dq}$:}
\begin{equation} \label{eq:vd_feedback}
    \mathbf{\dot{V}_\mathrm{dq}} = -k_v (\mathbf{V}_\mathrm{dq} -\overline{\mathbf{V}}_\mathrm{dq}), \\
\end{equation}
where $k_v$ is a positive proportional gain constant. It is easy to show that the closed-loop system is always stable (e.g., by computing the eigenvalues), and the trajectories remain within $\mathcal{I}$ with a correctly chosen $k_v$~\cite{joswig2024safe}.

\edit{If full information is not available, the following section demonstrates how our approach can be used to provide an optimal setpoint to existing grid-forming controllers that do not need to know the system parameters~\cite{gu2022power,huang2017transient}. Thus, our approach advances the grid-forming control technology by guaranteeing safe and optimal steady-state setpoints.}


\section{Case Study} \label{sec:simulation}
\

edit{To demonstrate our approach, we solve \eqref{eq:sdp} for both feasible and infeasible reference setpoints, then use different controllers to achieve the optimal setpoints. We evaluate the performance of a standard grid-forming droop controller adapted from \cite{9254645} given an infeasible setpoint against the same controller provided with an updated, optimal setpoint from \eqref{eq:sdp}. We then implement the linear controller in \eqref{eq:vd_feedback} to show that both transient and steady-state currents remain safe in the idealized case when full system information is available. }

\edit{To solve \eqref{eq:sdp} for different $(P^*,Q^*),(P^*,V_\mathrm{dq}^{2*})$ and $(Q^*,V_\mathrm{dq}^{2*})$ setpoints, we select the appropriate matrices $\bd M_1$, $\bd M_2$  and use the \texttt{CLARABEL} solver in \texttt{CVXPY} \cite{diamond2016cvxpy}. From the optimal solution $\bd W^*$, we extract $\overline{\mathbf{I}}_\mathrm{dq}^*$ and solve for the corresponding power and/or voltage equilibrium values that the controller needs to track.}

\edit{All simulations} are conducted using the \texttt{SciPy} odeint solver with a sampling time of $\Delta t = 100\mu s$ and simulation period of $t_{\mathrm{end}}= 1 s$. Reference setpoints used in the simulations are provided in Table \ref{tab:setpoints} where the disturbance is modeled as \edit{an update to the setpoint values} at time $t_0$ from pre- to post-disturbance values.  

\begin{table}
    \centering
    \caption{Simulation Parameters}
    \begin{tabular}{cc|cc}
    \hline
         Parameter & Value &  Droop Param. & Value \\ 
         \hline
         $R$& 0.8 $\Omega$ & $m_{p}$ & $2.6 \times 10^{-3}$ rad/sW \\
         $L$ & 1.5 mH & $m_{q}$ & $5.0 \times 10^{-3}$ V/VAr\\
         $E$ & 120 V & $m_{v^2}$ & 5.0 \\
         $\omega_\mathrm{nom}$ & $2\pi60$ rad/s & $\omega_{c}$ & $2\pi60$ rad/s\\
         \cline{3-4}
         $I_\mathrm{dq,nom}$ & 3.33 A & & \\
         $V_\mathrm{dq,nom}$ & 120 V & OC Param. & Value \\
         \cline{3-4}
         $S_\mathrm{inv,nom}$ & 1200 W & $k_{v}$ & 10 \\
         \hline
    \end{tabular}
    \label{tab:parameters}
\end{table}

\begin{table}
    \centering
    \caption{$P^*,Q^*,V^{2*}$ setpoints used in simulations}
    \begin{tabular}{cc|ccc}
        && $P^*$(W) & $Q^*$(VAr) & $V_\mathrm{dq}^{2*}$(V\textsuperscript{2})\\
        \hline
        $PQ$ & Pre-Disturbance & $800$ &  $0$ & --\\ 
         Tracking & Post-Disturbance & $\bf{1100}$ &  $0$ & -- \\
         \hline
        $PV^2$ & Pre-Disturbance & $200$ &  -- & $120^2$ \\
         Tracking & Post-Disturbance & $\bf{850}$ &  -- & $120^2$  \\
        \hline
        $QV^2$ & Pre-Disturbance & -- &  $0$ & $120^2$  \\
        Tracking & Post-Disturbance & -- &  $\bf{-500}$ & $120^2$  \\
        \hline
    \end{tabular}
    \label{tab:setpoints}
\end{table}

\edit{\subsection{Grid-Forming Droop Control Comparison}}
The $PQ$ droop controller first filters the inverter's active and reactive power ($P$,$Q$) through a low-pass filter with cut-off frequency $\omega_c$ to measure $\tilde{P}$ and $\tilde{Q}$. These filtered values have dynamics $ \dot{\tilde{P}} =\omega_c(P-\tilde{P}), \; \dot{\tilde{Q}} =\omega_c(Q-\tilde{Q}).$

Then, voltage magnitude droop and frequency droop are implemented as $\dot{V}_\mathrm{dq} = m_q\omega_c(\tilde{Q}-Q)$ and $\Delta \omega_i = -m_p(\tilde{P}-P^*)$, where $\Delta \omega_i$ is the deviation of inverter frequency from grid frequency, and $m_p$ and $m_q$ are proportional gain parameters. These follow from $\Delta V_\mathrm{dq}=-m_q(\tilde{Q}-Q^*)$, where $\Delta V_\mathrm{dq}$ is the deviation in voltage magnitude from a desired $V_\mathrm{dq,nom}$ value.

The $PV^2$ droop controller has the same low-pass $PQ$ filter and power-frequency dynamics as the $PQ$ droop controller, but it replaces the reactive power-voltage dynamics  with linear feedback on $V^2_\mathrm{dq}$ as $\dot{V}^2_\mathrm{dq} = -m_{v^2}(V^2_\mathrm{dq} - V^{2*}_\mathrm{dq}).$ Parameters for the $PQ$ and $PV^2$ droop controllers are listed in Table \ref{tab:parameters} (adapted from \cite{9254645}).

\begin{figure}[ht]
\centering
\includegraphics[width=.7\linewidth]{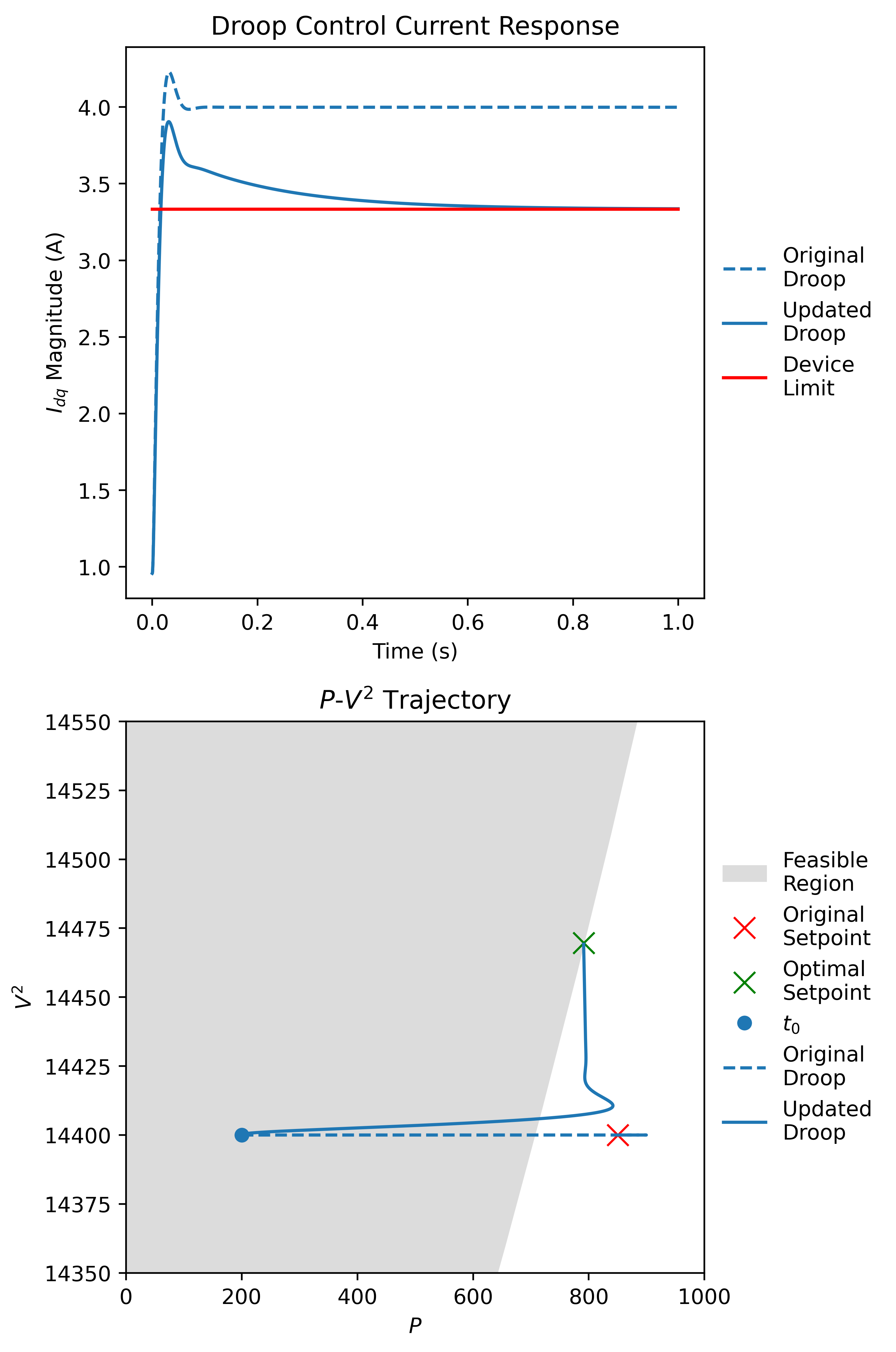}
\caption{\edit{The two plots compare the current magnitude response (top) and $(P,V^{2})$ region trajectory (bottom) for the same droop controller tracking two different setpoints: the original infeasible setpoint (dashed blue line) and the best feasible setpoint output from the optimization problem (solid blue line). Providing an updated, optimized setpoint to the droop controller instead of the original reference setpoint ensures it achieves a safe steady-state value regardless of whether the original setpoint was within the feasible operating region. The transient response violates the current constraints for a short time, which might be tolerable~\cite{sorensen2013thermal,firouz2014efficiency} or can be mitigated with other techniques~\cite{qoria2020current,fan2022review}.}}
\label{fig:droop_pv2}
\end{figure}

\edit{Figure \ref{fig:droop_pv2} shows the unsafe transient \emph{and} steady-state response of the original $PV^2$ droop controller to an infeasible setpoint of $P^*=850W,V_\mathrm{dq}^{2*}=120^2V$. In contrast, when we solve \eqref{eq:sdp} to provide an updated, safe $(P^*,V_\mathrm{dq}^{2*})$ setpoint for the droop controller, we allow the grid-forming inverter to settle to the optimal feasible operating point. This prevents a consistent and unsafe violation of current limits in steady-state.}

\edit{\subsection{Linear Feedback Optimal Controller (OC)}}
\begin{figure}[ht]
\centering
\begin{minipage}{0.7\linewidth}
    \includegraphics[width=\linewidth]{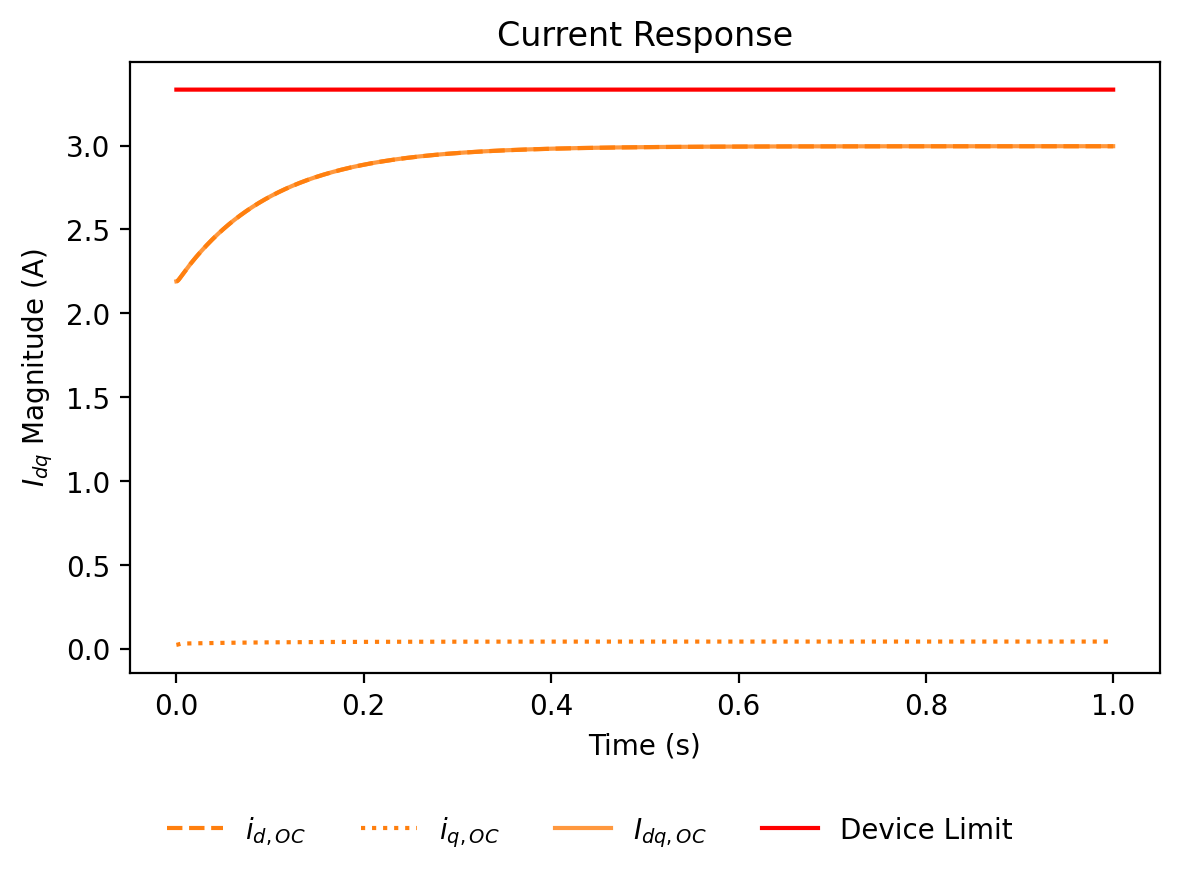}
\end{minipage} \\
\begin{minipage}{0.7\linewidth}
    \includegraphics[width=\linewidth]{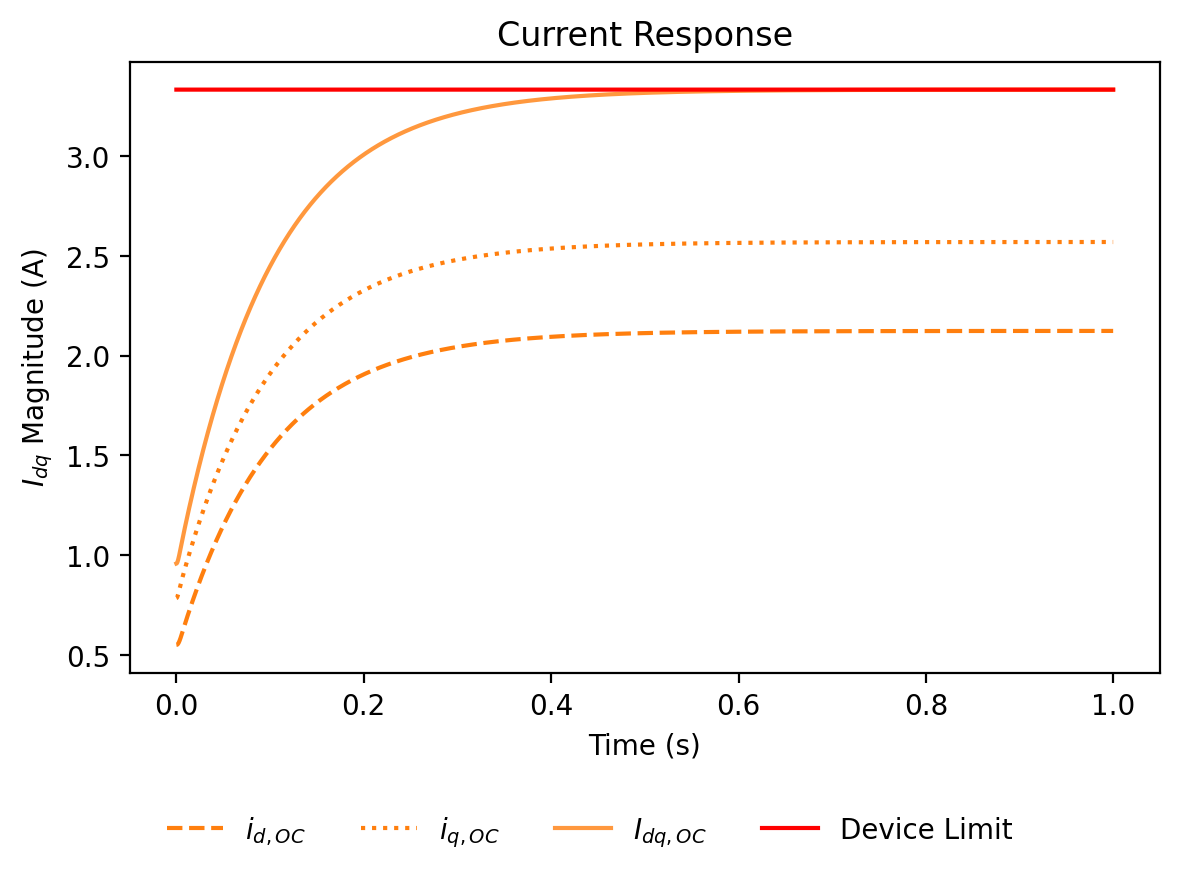}
\end{minipage}
\caption{\edit{The top plot is the response of the optimal controller current magnitude to a feasible $(P^*,Q^*)$ setpoint and the bottom plot is the response of the optimal controller to an infeasible $(P^*,V^{2*})$ setpoint. In both simulations, the inverter current magnitude remains within the safe operating region. The infeasible setpoint in the bottom plot causes the inverter current to settle at the boundary of the safe operating region.}}
\label{fig:oc_response}
\end{figure}
\edit{The linear feedback controller in \eqref{eq:vd_feedback} tracks $\oVdq^* = \mathbf{E}_\mathrm{dq} - L\mathbf{A}\oIdq^*$. We implement feedback on $\bd V_\mathrm{dq}$ with the value of $k_v$ listed in Table \ref{tab:parameters}.}
\edit{We simulate two simple examples of the optimal controller tracking $PQ$ and $PV^2$ setpoints. Figure \ref{fig:oc_response} shows the inverter current response following a feasible change in the $PQ$ tracker's setpoint and the current response following an infeasible change in the $PV^2$ tracker's setpoint (see Table \ref{tab:setpoints} for details). The linear feedback controller ensures both the transient and steady-state currents remain within the safe operating region of the inverter. }


\edit{These advantages of the optimal controller motivate further investigation into whether we can control the inverter to achieve safety at all times without knowing the grid voltage magnitude or the $RL$ filter values. We believe our methods can be adapted to the case where the parameters are not known, and we address this in a follow-up work.}

\section{Conclusion} \label{sec:conclusion}
In this paper, we studied the geometry of the feasible output region of a current-limited inverter. We showed that this region is convex and can be described using linear matrix inequalities. \edit{We demonstrated how to use this fact to improve grid-forming controllers such that steady-state currents remain within the current magnitude limit. Some future directions include showing the semidefinite program always returns a rank 1 solution, and how the controllers can be designed when the detailed parameters are not known.}



\bibliographystyle{IEEEtran}
\bibliography{Reference}
\end{document}